\documentclass[preprint, 11pt]{article}
\usepackage[utf8]{inputenc}
\input{epsf}
\usepackage[utf8]{inputenc}
\usepackage{amsmath}
\usepackage{amssymb}
\input{epsf}
\usepackage{natbib}
\usepackage[margin=1in]{geometry}
\usepackage[raggedright]{titlesec}
\usepackage[hang,flushmargin]{footmisc} 
\usepackage{graphicx}
\usepackage{framed}
\usepackage{color}
\usepackage{yfonts}
\usepackage[dvipsnames]{xcolor}
\usepackage{pifont}
\usepackage{amsmath,amssymb}
\usepackage{multirow}
\usepackage{amsfonts}
\usepackage{subfigure}
\usepackage{caption}
\usepackage{threeparttable}
\usepackage{esvect}
\usepackage{pgf}
\usepackage{tikz}
\usepackage{tikz-cd}
\usepackage{fancyvrb}
\usetikzlibrary{positioning}
\usepackage{bbm}
\usepackage{colortbl}
\usepackage{pdflscape}
\usepackage{booktabs}
\usepackage{amsmath}
\usepackage{mathdots}
\usepackage{yhmath}
\usepackage{cancel}
\usepackage{siunitx}
\usepackage{array}
\usepackage{multirow}
\usepackage{gensymb}
\usepackage{tabularx}
\usepackage{booktabs}
\usetikzlibrary{fadings}
\usetikzlibrary{patterns}
\usetikzlibrary{shadows.blur}
\usetikzlibrary{shapes}
\usepackage{placeins}

\usepackage{hyperref}
\hypersetup{
    colorlinks=true,
    linkcolor=NavyBlue,
    filecolor=magenta,      
    urlcolor=cyan,
    citecolor=black
}
 
\usepackage{rotating}
\usepackage{ntheorem}
\makeatletter
\renewtheoremstyle{plain} 
{\item{\theorem@headerfont ##1\ ##2\theorem@separator}~}
{\item{\theorem@headerfont ##1\ ##2\ (##3)\theorem@separator}~}
\makeatother

\theoremheaderfont{\normalfont\bfseries}

\newtheorem{theorem}{Theorem}[section]
\newtheorem{definition}{Definition}[section]
\newtheorem{lemma}{Lemma}[section]

\newtheorem{proposition}{Proposition}[section]

\newtheorem{assumption}{Assumption}

\newenvironment{proof}{\paragraph{Proof:}}{\hfill$\square$}

\newcommand{\E}{\mathbb{E}}

\newcommand{\var}{\text{var}}
\newcommand{\cov}{\text{cov}}
\newcommand{\cor}{\text{cor}}

\newcommand{\bX}{\mathbf{X}}

\newcommand{\cP}{\mathcal{P}}
\newcommand{\cS}{\mathcal{S}}

\newcommand{\indep}{\raisebox{0.05em}{\rotatebox[origin=c]{90}{$\models$}}}
\definecolor{shadecolor}{gray}{0.9}

\tikzset{every picture/.style={line width=0.75pt}} 

\usepackage{enumitem}
\newlist{Step}{enumerate}{2}
\setlist[Step]{label={{Step \arabic*.}}, leftmargin=*}

\usepackage{setspace}

 \makeatletter
\newcommand\circled[1]{%
  \mathpalette\@circled{#1}%
}
\newcommand\@circled[2]{%
  \tikz[baseline=(math.base)] \node[draw,circle,inner sep=2pt] (math) {$\m@th#1#2$};%
}
\makeatother

 \makeatletter
\newcommand\circledblue[1]{%
  \mathpalette\@circledblue{#1}%
}
\newcommand\@circledblue[2]{%
  \tikz[baseline=(math.base)] \node[draw,circle, fill=blue!20, inner sep=2pt] (math) {$\m@th#1#2$};%
 }

\renewenvironment{abstract}
 {\begin{center}\normalsize\textsc{Abstract}%
 \end{center}\begin{quote}\normalsize}
 {\end{quote}}

\title{Overlap Violations in External Validity}
\author{Melody Huang\thanks{Postdoctoral Fellow, Harvard University, Email: \texttt{melodyhuang@fas.harvard.edu}. \newline The author would like to thank Erin Hartman, Anna Wilke, Harsh Parikh, Lauren Liao, Yaxuan Huang, Diana Lee, and the Miratrix C.A.R.E.S. Lab for their helpful comments and feedback. The author would also like to thank Chris Blattman for generously providing access to the data.}}
\date{}

\begin{document}
\maketitle 
\begin{abstract} 
Estimating externally valid causal effects is a foundational problem in the social and biomedical sciences. Generalizing or transporting causal estimates from an experimental sample to a target population of interest relies on an overlap assumption between the experimental sample and the target population--i.e., all units in the target population must have a non-zero probability of being included in the experiment. In practice, having full overlap between an experimental sample and a target population can be implausible. In the following paper, we introduce a framework for considering external validity in the presence of overlap violations. We introduce a novel bias decomposition, that parameterizes the bias from an overlap violation into two components: (1) the proportion of units omitted, and (2) the degree to which omitting the units moderates the treatment effect. The bias decomposition offers an intuitive and straightforward approach to conducting sensitivity analysis to assess robustness to overlap violations. Furthermore, we introduce a suite of sensitivity tools in the form of summary measures and benchmarking, which help researchers consider the plausibility of the overlap violations. We apply the proposed framework on an experiment evaluating the impact of a cash transfer program in Northern Uganda. 
\end{abstract}

\newpage 
\doublespacing
\section{Introduction}
The credibility revolution has pushed for careful causal identification. Randomized control trials have come to be viewed as the gold standard for estimating causal effects. Within a given experiment, researchers can randomly assign treatment, allowing them to estimate the average treatment effect across an experimental sample with minimal identifying assumptions. However, while experiments provide researchers with the ability to estimate internally valid effects, there is an underlying question as to whether or not the estimated effect is also externally valid. Being able to generalize or transport causal findings beyond an experimental setting is foundational to answering broader research questions. 

In practice, the experimental sample is rarely a representative sample from the target population of interest. Existing literature has outlined the necessary assumptions for estimating externally valid causal effects \citep[e.g.,][]{egami2020elements, kern2016assessing, stuart2011use, imai2008misunderstandings}. To generalize or transport the results from an experimental sample to a target population of interest, researchers often rely on two key assumptions. The first is in the form of a conditional ignorability assumption, in which researchers must assume that they are able to measure a sufficiently rich set of moderators that can account for the confounding effects of selection into the experimental sample (e.g., \citealp{imai2008misunderstandings, cole2010generalizing, olsen2013external}). The second is an overlap assumption, in which researchers assume that there is a non-zero probability that units in the target population can be included in the experimental sample \citep[e.g.,][]{rosenbaum1983assessing, tipton2014generalizable}. 

While recent literature has introduced a variety of different methods and sensitivity analyses to assess robustness to violations of conditional ignorability \citep[e.g.,][]{huang2022sensitivity, nie2021covariate, colnet2021generalizing, dahabreh2019sensitivity, nguyen2017sensitivity}, little work has been done to consider overlap violations in the context of external validity. However, overlap violations are a foundational concern when generalizing or transporting causal effects. In practice, subsets of the target population can have zero probability of being included in the experimental sample. This could occur due to units being difficult to recruit or reach (as is the case in generalizability settings), or as a result of the experiment and the target population being in two different geographic locations (as is the case in transportability settings). 

The goal of this paper is to introduce a framework for researchers to consider the sensitivity in their underlying estimates to overlap violations in external validity. The paper provides three primary contributions. First, we build on the work from \cite{egami2020elements} to clarify what constitutes an overlap violation. We introduce the notion of \textit{transportable populations}, which represents the set of all units that the experimental sample could feasibly be transported to, and formalize that an overlap violation occurs when the target population of interest is not a subset of the transportable population. 

Second, we decompose the bias from an overlap violation into two different components: (1) the proportion of units missing from the transportable population in the target population, and (2) the degree to which omitting units from the transportable population moderates the treatment effect. These two factors characterize the different dimensions of an overlap violation, and provide insight into understanding ways to mitigate overlap violations at the recruitment stage of an experiment. Furthermore, both parameters are standardized and can be upper bounded. 

Finally, we introduce a suite of tools that help researchers calibrate the plausibility of certain overlap violations for a given experiment. We propose numerical summary measures that help quantify the minimum amount of moderation that an omitted subgroup must explain in order for an overlap violation to overturn a research result. Furthermore, we introduce a benchmarking approach, which allows researchers to estimate the bias that would occur from omitting units from the target population, similar to omitting an observed subgroup in the data. These tools allow researchers to incorporate their substantive expertise into the sensitivity framework, and provides a transparent way to report potential sensitivity to overlap violations in a given study.

An outline of the paper follows. In Section 2, we introduce the notation and necessary assumptions for generalizing or transporting causal effects, as well as the running example. Section 3 formalizes what it means for an overlap violation to occur, and introduces a bias decomposition for researchers to understand the different drivers of error from overlap violations. In Section 4, we discuss the different ways in which researchers can calibrate their understanding of plausible overlap violations. Section 5 concludes. Proofs, extensions, and additional discussion are provided in the Appendix. 
\section{Background}
\subsection{Notation and Assumptions} 
Following \cite{buchanan2018generalizing}, we define an experimental sample as an i.i.d. sample of $n$ units, drawn potentially with bias from an infinite super-population. Let $S_i$ be a selection indicator, which takes on a value of 1 when a unit is included in the experimental sample, and 0 otherwise. We define the target population as an i.i.d. sample of $N$ units, drawn randomly with equiprobability from an infinite super-population. Let $\cP$ be defined as the set of indices that correspond to the units in the target population, and $\cS$ as the set of indices corresponding to the units in the experimental sample. 

Let $T_i$ be a binary treatment assignment indicator, where $T_i = 1$ if a unit is assigned to treatment, and 0 otherwise. Throughout the paper, we will assume that treatment is randomly assigned, but the framework can be easily extended for observational settings. Let $Y_i(t)$ represent the potential outcome for unit $i$ under treatment assignment $t \in \{0,1\}$. We will make the standard assumptions of no interference and that treatment is identically administered across all units (i.e., SUTVA), and assume full compliance, such that the observed outcomes $Y_i$ can be written as $Y_i := Y_i(1) T_i + Y_i(0) (1-T_i)$.  

The estimand of interest in this setting is the \textit{target population average treatment effect} (T-PATE): 
$$\tau := \E\big[ Y_i(1) - Y_i(0) \mid S_i = 0\big],$$
which represents the average treatment effect across the target population (i.e., $S_i = 0$). In settings when the experimental sample is a random draw from the infinite super-population, then the difference-in-means estimator can be used as an unbiased estimator for the T-PATE. However, when this is not the case, we must leverage two additional assumptions to identify the T-PATE \citep[e.g.,][]{imai2008misunderstandings, cole2010generalizing, olsen2013external}. The first is a conditional ignorability assumption, which states that there exists a set of moderators--i.e., covariates that drive treatment effect heterogeneity--that can explain that the confounding effects from selection into the experimental sample.
\begin{assumption}[Conditional Ignorability of Treatment Effect Heterogeneity] \label{assum:cond_ign} \mbox{}\\
$$Y_i(1) - Y_i(0) \ \indep \ S_i \mid \mathcal{X}_i$$
where $\mathcal{X}$ represents a set of pre-treatment covariates.
\end{assumption} 

Furthermore, we must also invoke an overlap assumption \citep{rosenbaum1983assessing, westreich2010invited}. 

\begin{assumption}[Overlap] \label{assum:positivity} 
$$0 < P(S_i = 1 \mid \mathcal{X}_i) \leq 1$$
\end{assumption} 

Overlap assumes that conditional on the covariate profiles $\mathcal{X}_i$, all units in the target population have a non-zero chance of being included in the experimental sample.

A common approach to estimating the T-PATE in practice is to use weighted estimators \citep{stuart2011use, olsen2013external, buchanan2018generalizing}:
$$\hat \tau = \frac{1}{n_1} \sum_{i:S_i = 1} \hat w_i Y_i T_i - \frac{1}{n_0} \sum_{i:S_i = 1} \hat w_i Y_i (1-T_i).$$
Common examples of weights that are used in practice include inverse propensity score weights \citep[e.g.,][]{cole2010generalizing, stuart2011use, buchanan2018generalizing}, or balancing weights, which allow researchers to bypass modeling the underlying propensity of selection into the experimental sample, and directly target the distributional difference in the covariate distributions \citep[e.g.,][]{sarndal2003model, hainmueller2012entropy, josey2021transporting, lu2021you, ben2020balancing}. 

Researchers can alternatively choose to model the treatment effect heterogeneity by estimating a model for the individual-level treatment effect across the experimental sample, and using the model to project into the target population \citep{kern2016assessing}. Recent work has also introduced augmented weighted estimators, which allow researchers to simultaneously model both processes, with the advantage of doubly robustness \citep{dahabreh2019generalizing}.

Regardless of which estimation approach researchers choose to employ, to obtain an unbiased estimate of the T-PATE, both Assumption \ref{assum:cond_ign} and \ref{assum:positivity} must hold. While existing literature has introduced different approaches to consider the potential violations to Assumption \ref{assum:cond_ign} \citep[e.g.,][]{huang2022sensitivity, dahabreh2019sensitivity, nguyen2017sensitivity, nie2021covariate}, little work has been done to consider potential sensitivity to violations of overlap. In practice, overlap violations (also referred to as undercoverage, or positivity violations) arise in almost every experimental study \citep{tipton2013improving, tipton2023designing}. For example, medical trials often have an eligibility criteria to participate to minimize potential co-morbidities \citep[e.g.,][]{britton1999threats}. However, this also means that there will be subsets of individuals in the target population who are systematically not represented in experimental studies. Similarly, when researchers wish to transport the results from an experimental site to a separate, target population, contextual differences between the target population and the experimental sample threaten the external validity of the study. Inconveniently, overlap violations cannot be overcome with post-hoc adjustments or collecting more covariate data; as such, violations in overlap can be even more pernicious than violations of conditional ignorability.  

Recent work in the literature has primarily focused on introducing data-driven approaches to re-define the target population to comprise of individuals represented in the experimental sample. For example, estimating overlap weights \citep[e.g.,][]{cheng2022addressing} or trimming \citep[e.g.,][]{crump2009dealing} allows researchers to generalize (or transport) their experimental results to a subset of the target population that comprises of units similar to those in the experiment. While these approaches allow researchers to generalize their results to the set of units in the target population for which overlap holds, it requires altering the underlying estimand of interest from the T-PATE \citep[e.g.,][]{parikh2024we}. 

The goal of this paper is thus to introduce a sensitivity framework to assess how sensitive estimates of the T-PATE are to potential bias from overlap violations. This allows researchers to consider the full target population of interest, without having to re-define their estimand of interest. We will focus explicitly on weighted estimators, though the framework can be flexibly applied, regardless of what estimation approach researchers choose to employ (see Appendix \ref{app:extra_details} for more details). 

\subsection{Motivating Example: Cash Transfer Programs}
To help motivate and illustrate the proposed framework, we will turn to a cash transfer experiment, conducted by \citet{blattman2014generating} across Northern Uganda. The Youth Opportunities Program (YOP) was a large scale development program targeting poor and unemployed young adults in Northern Uganda from 2008 to 2012. In 2008, the government solicited grant proposals for either aid in starting a business or vocational training. Treatment was then randomly assigned among the applications that made it past the screening phases. If units were assigned to treatment, they would receive a one-time, unconditional grant, averaging around \$7,500 per group. Units in the control group were not given any cash. Follow-up surveys were then conducted two years later to measure a variety of different outcomes, such as cash earnings, business profits, hours of vocational training, political participation, etc. The authors found that individuals who received an unconditional cash transfer were 4.5 times more likely to have vocational training and had 42\% higher earnings. Given the strong and persistent positive effects, the authors concluded that the intervention of an unconditional cash transfer was largely a success.

A natural policy question that arises is whether or not the impact of such an intervention would also be as effective had it been applied outside of the experimental sample. In other words, would the estimated impact of the cash transfer generalize to the rest of Northern Uganda? For simplicity of illustration, we will focus on generalizing the impact of two of the outcome measures: hours of vocational training and cash earnings. We define our target population to be the rest of Northern Uganda, using a population-based household survey (Northern Uganda Survey--i.e., NUS). To account for potential moderators, we weight on a set of pre-treatment covariates measured across both the experimental sample and the target population \citep{egami2019covariate}, such as age, gender, education, and durable assets. Within the experimental sample, researchers estimated an increase of 337 hours of vocational training amongst units who received the cash transfer (i.e., around 8-9 times as many hours as those in the control group). Furthermore, individuals in the treatment group made 14,000 more UGX than individuals in the control group. Generalizing these impacts to the rest of Northern Uganda, we see both effects attenuated towards zero, with an estimated increase of 207 hours of vocational training (i.e., around 5-6 times as many hours as those in the control group) and an increase of 7,220 UGX.

\begin{table}[ht]
\centering
\begin{tabular}{lcc}
  \toprule 
Outcome & Within-Site & Weighted  \\ 
  \midrule 
Hours of Vocational Training & 337.46 (16.28) & 207.16 (68.37) \\ 
Cash Earnings$^*$ & 14.00 (3.76) & 7.22 (7.14) \\ 
   \bottomrule
   \multicolumn{3}{l}{\footnotesize  $^*$-000s of UGX (2008); $n=2,005$, $N = 18,041$}
\end{tabular}
\end{table}

While we have estimated the T-PATE using the observed covariate data, a concern that remains is the strong selection effect. More concretely, the experimental sample comprises only of individuals who could have filled out the application and passed a screening phase. In contrast, the target population contains a mixture of individuals who may or may not have the ability or initiative to fill out the grant application, and individuals whose applications would not make it through the screening. As such, there exists an overlap violation, and we will not be able to correctly recover the true T-PATE. This will be true, \textit{even if} researchers had been able to measure all of the latent characteristics that could moderate the treatment effect--i.e., unobserved initiative, familial connections, affinity for entrepreneurship, etc. 
Throughout the paper, we will apply our proposed sensitivity framework to the example to illustrate how researchers can assess sensitivity to overlap violations in their T-PATE estimation. 

\section{Bias from Overlap Violations} \label{sec:bias}

 \subsection{Formalizing Transportable Populations}
In the following subsection, we will formalize what constitutes an overlap violation. We begin by defining \textit{transportable populations}.

\begin{definition}[Transportable Populations] \label{def:transport_pop} \mbox{}\\
The transportable population is defined as the set of all units $i$ in the target population, for which the probability of inclusion, conditional on the minimum separating set of moderators $\mathcal{X}$ in the experimental sample is greater than zero: 
$$\phi_\cS := \{ i \in \mathcal{P}: P(S_i = 1 \mid \mathcal{X}_i) > 0\}.$$
\end{definition} 
Definition \ref{def:transport_pop} clarifies that in all settings, for a given experimental sample, there exists \textit{some} population for which the results may be transported or generalized to. In particular, given $\phi_\cS$, under Assumption \ref{assum:cond_ign} (i.e., $\mathcal{X}$ is fully observed and measured), the experimental sample can always be transported or generalized to recover the average treatment effect, across the \textit{transportable population}: 
$$\tau_{\phi} = \E \left[ Y_i(1) - Y_i(0) \mid S_i=0, i \in \phi_\cS \right].$$
In the context of the running example, the transportable population may consist of all individuals across Northern Uganda who has similar levels of initiative, familial connections, and affinity for entrepreneurship as the units represented in the experimental sample. In other words, the transportable population would not comprise of individuals who were less motivated, or less well-connected. The transportable population is the modified estimand that trimmed estimators or overlap weights will target. 

The set of transportable populations does not depend on researchers fully measuring the set of variables $\mathcal{X}$, though recovering $\tau_\phi$ does. In particular, whether or not we have actually collected information about individual's initiative or familial connections does not affect which population the experimental sample can be feasibly generalized (or transported) to. Overlap violations thus do not depend on the measured covariates, and occur when there is a divergence between the set of feasible transportable populations and the actual target population of interest. The following definition formalizes.

\begin{definition}[Overlap Violation] \label{def:overlap} \mbox{}\\
An overlap violation occurs if the target population $\cP$ is not equal to the transportable populations: $\mathcal{P} \neq \phi_\cS$.
\end{definition} 
In other words, we define an overlap violation as any setting in which there are units in the target population for which the conditional probability of inclusion, given the necessary set of moderators, is zero.\footnote{Recent work in the observational causal inference literature has considered violations of a stronger assumption, known as strong overlap, in propensity scores \citep[e.g.,][]{lei2021distribution}. Strong overlap requires that the probability of selection is strictly bounded by a constant (i.e., there must exist some $0 < \eta \leq 0.5$, where for all $x \in \mathcal{X}$, $\eta < \Pr(S = 1 \mid X = x) < 1-\eta$). While an overlap violation as defined in Definition \ref{def:overlap} implies a violation in strong overlap, the converse is not true. Thus, we view the framework introduced in this paper as distinct from the existing literature on violations in strong overlap.}

Overlap violations encapsulate many concerns that arise in external validity. We detail a few examples below. 

\paragraph{Example: Contextual Shifts.} A common concern in external validity is the presence of contextual shifts between the experimental sample and the target population (i.e., $C$-validity, in the language of \citealp{egami2020elements}). Contextual shifts can occur as a result of changes in geographical regions, institutional differences, time, etc. Recent work by \cite{egami2020elements} formalized that in settings when researchers are concerned about contextual shifts between the experimental sample and the target population, they must leverage a \textit{contextual exclusion restriction} to control for moderators that account for the treatment effect heterogeneity that arise as a result of contextual differences. However, the contextual restriction is violated in settings when a contextual moderator is constant for all units in the experimental sample. Consider the running example. Assume \textit{motivation} is an important moderator that accounts for contextual differences between the experimental sample and the target population in the cash transfer setting. However, everyone who selects into the experimental sample was required to fill out a grant application, and as such, is inherently motivated. Even if researchers were able to measure the latent trait of motivation, there is an overlap problem: the experimental sample only comprises of individuals above a certain baseline level of motivation, whereas individuals in the target population may have varying levels of motivation. As such, the contextual exclusion restriction would fail to hold.\footnote{Overlap violations constitute one type of violation in the contextual exclusion restriction. The contextual exclusion restriction could also be violated in settings when researchers fail to collect the full set of moderators needed. See \citet{huang2022sensitivity} for more discussion on such settings. }

\paragraph{Example: Attrition.} Another common concern when considering the external validity of experiments is attrition. Attrition can also be thought of as an overlap violation. In particular, individuals who leave the experiment are no longer included in the experimental sample. However, if the treatment is going to be deployed across a larger target population, then the set $\phi_\cS$ would exclude units who would drop out of treatment. 

\subsection{Bias Decomposition} 
In this subsection, we decompose the bias from an overlap violation as a function of the proportion of omitted units and an $R^2$ value. To begin, let the minimum separating set of moderators (i.e., the smallest set of moderators needed for Assumption \ref{assum:cond_ign} to hold) be defined as $\mathcal{X} := \{ X, V\}$.  Let $V$ be a binary indicator, where without loss of generality, $V_i = 0$ for all units in the experimental sample, but is a mixture of $\{0,1\}$ in the target population. We will assume that  $V_i \indep X_i$.\footnote{The assumption that $V_i \indep X_i$ is akin to \citet{rosenbaum1983assessing}'s consideration of the \textit{residual} contribution of an omitted variable \citep[e.g.,][]{cinelli2020making, nguyen2017sensitivity, huang2022vbm, huang2023design, hartman2022sensitivity}.} In the context of Definition \ref{def:transport_pop}, the transportable population can then be written as the set of units in the target population, where $V_i = 0$ (i.e., $\phi_\cS := \{ i \in \mathcal{P}: V_i = 0\}$). 

We are interested in decomposing the difference between the average treatment effect across the target population (i.e., $\E(\tau_i \mid S_i = 0)$) and the average treatment effect across the transportable population (i.e., $\E(\tau_i \mid S_i = 0, V_i = 0)$). To do so, we linearly decompose the individual-level treatment effect $\tau_i$ across the target population: 
\begin{equation} 
\tau_i = \hat \alpha + \hat \gamma V_i + u_i \mid S_i = 0,
\label{eqn:tau_lin_decomp}
\end{equation}
where $\hat \alpha := \E(\tau_i \mid V_i = 0, S_i = 0)$ (i.e., the average treatment effect across the transportable population), and $u_i$ is a residual noise term (i.e., $u_i := \tau_i - \hat \alpha - \hat \gamma V_i$). $\hat \gamma$ represents the difference between $\E(\tau_i \mid S_i = 0, V_i = 1)$ and $\E(\tau_i \mid S_i = 0, V_i = 0)$. 

Notably, the linear decomposition in Equation \eqref{eqn:tau_lin_decomp} can be performed without loss of generality--i.e., we are not assuming that the underlying data generating process is linear. The linear decomposition is helpful to re-formulate the bias from an overlap violation in terms of the underlying coefficients. Using Equation \eqref{eqn:tau_lin_decomp}, we can write the bias as: 
\begin{align*} 
\text{Bias}(\hat \tau) &= \E(\tau_i \mid S_i = 0, V_i = 0) - \E(\tau_i \mid S_i = 0) \\
&= \hat \gamma \cdot P(V_i = 1 \mid S_i = 0).
\end{align*} 
We can then leverage the closed-form representation of $\hat \gamma$ to re-parameterize the bias in terms of an $R^2$ value. The following theorem formalizes.

\begin{theorem}[Bias from Overlap Violations] \label{thm:bias} Assume Assumption \ref{assum:cond_ign} (Conditional Ignorability) holds for the set of observed covariates $X_i$ and a binary indicator $V_i$. Assume that $V_i = 0$ for all units in the experimental sample, while the target population comprises of units with $V_i = 0$ and $V_i = 1$. Then, the bias from omitting the units $V_i = 1$ from the experimental sample can be written as follows: 
 \begin{align} 
\text{Bias}(\hat \tau) &= \sqrt{R^2_{\tau \sim V} \cdot \var(\tau_i \mid S_i = 0) \cdot \frac{p}{1-p}} \label{eqn:bias_vartau} \\ 
&=\sqrt{\frac{R^2_{\tau \sim V}}{1-R^2_{\tau \sim V}} \left(1+C_\sigma  \cdot \frac{p}{1-p}\right) p \cdot \var_w(\tau_i \mid S_i = 1)}, \label{eqn:bias}
\end{align} 
where $\var_w(\cdot)$ is defined as a weighted variance, $R^2_{\tau \sim V} := \cor(\tau_i, V_i \mid S_i = 0)^2$, $p := \E(V_i \mid S_i = 0)$, and $C_\sigma := \var(u_i \mid V_i = 1, S_i = 0)/\var(u_i \mid V_i = 0, S_i = 0)$ is a hyperparameter that controls for unobserved heteroskedasticity.
\end{theorem} 
Equation \eqref{eqn:bias_vartau} in Theorem \ref{thm:bias} highlights that bias from an overlap violation will be driven by three factors: (1) proportion of units omitted (i.e., $p$), (2) how strongly the omission moderates the treatment effect (i.e., $R^2_{\tau \sim V}$), and (3) how much \textit{total} treatment effect heterogeneity there is in the target population. 

In practice, researchers do not know how much total treatment effect heterogeneity there is in the target population (i.e., $\var(\tau_i \mid S_i = 0)$). This is exacerbated by the fact that in the presence of an overlap violation, by omitting units from the experimental sample, we are omitting subsets of the target population for which we are concerned there will be \textit{additional} heterogeneity. Thus, the treatment effect heterogeneity from the experimental sample can, at best, serve as a lower bound for the treatment effect heterogeneity across the target population. Instead, Equation \eqref{eqn:bias} from Theorem \ref{thm:bias} utilizes the fact that we can decompose the total variation in the individual-level treatment effect across the target population as a function of the treatment effect heterogeneity captured in the experimental sample (i.e., $\var_w(\tau_i \mid S_i = 1)$), the parameters $\{R^2_{\tau \sim V}, p\}$, and a hyperparameter $C_\sigma$. Section \ref{subsec:var_tau} provides additional discussion. 

$C_\sigma$ controls for how much additional unexplained heterogeneity there is across the units omitted in the target population (i.e., $V_i = 1$) and the units included (i.e., $V_i = 0$). If $C_\sigma \approx 1$, then this implies that while there may be differences in the treatment effect across $V_i = 0$ and $V_i=1$ units, the amount of variation across both groups is roughly the same. If $C_\sigma > 1$, then this implies that there is more variation across the units that were omitted (i.e., $V_i = 1$) than the units included. For example, in a randomized control trial for a new drug, researchers may be concerned that there is greater variation in how individuals not represented in the medical trial respond to the drug. If $C_\sigma < 1$, this implies that there is less treatment effect heterogeneity across the omitted units than in the included units. For example, in the cash transfer context, we may expect that individuals who are still in school may have a relatively constant impact from receiving treatment on the hours of vocational training, as the number of available hours after attending school is limited. As a result, if researchers are concerned about overlap violations with respect to the subset of the population that is too young to be eligible for the program, we might expect there to be \textit{less} heterogeneity in the treatment effect across the excluded $V_i=1$ units. 

\begin{figure}[t]
\includegraphics[width=\textwidth]{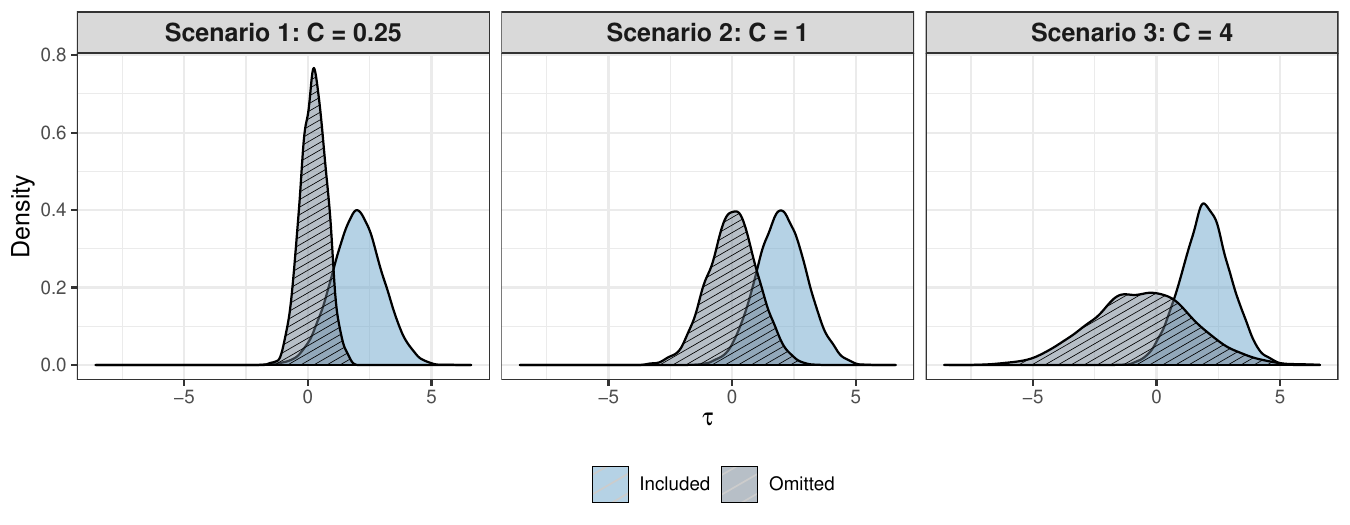}
\caption{Simulated examples comparing the distribution of $\tau_i$ across included ($V_i=0$) and excluded $(V_i=1)$ units. Both $p$ and $R^2_{\tau \sim V}$ are fixed in all three scenarios.}
\label{fig:c_sigma_example}
\end{figure} 

In general, we recommend researchers set $C_\sigma = 1$, and then repeat the analysis for varying values of $C_\sigma$. We provide more discussion on how researchers can calibrate potential values of $C_\sigma$ in Appendix \ref{app:c_sigma}. In general, as we increase $C_\sigma$, this implies there is a greater gap between the total treatment effect heterogeneity ($\var(\tau_i \mid S_i = 0)$) and the treatment effect heterogeneity captured by the experimental sample ($\var_w(\tau_i \mid S_i = 1)$). Figure \ref{fig:c_sigma_example} visualizes simulated distributions of $\tau_i$ with varying $C_\sigma$.

We can apply a Fréchet-Hoeffding bound to construct a sharp upper bound on $\var_w(\tau_i \mid S_i = 1)$ using the observed, empirical cumulative distribution functions of the outcomes (see \citealp{ding2019decomposing} and \citealp{huang2022sensitivity} for more discussion). As a result, the only unobserved parameters in Equation \eqref{eqn:bias} are $R^2_{\tau \sim V}$ and $p$. Thus, to assess the bias from an overlap violation, we can treat $\{R^2_{\tau \sim V}, p\}$ as our sensitivity parameters. 

It is worth noting that Theorem \ref{thm:bias} is not the only way in which researchers can choose to decompose the bias from an overlap violation. For example, \citet{manski1989anatomy} propose assumption-free bounding approaches for attrition, which can be straightforwardly extended for overlap violations (see also \citealp{gerber2012field}). However, there are several advantages to the bias decomposition in Theorem \ref{thm:bias} over alternative approaches. First, it provides a framework to interpret and reason about bias  using standardized parameters, and also allows for a richer way for researchers to encode their substantive priors into the sensitivity parameter (see Section \ref{sec:tools} for discussion). Second, because bounding approaches are completely assumption-free, the resulting range of possible estimates tend to be extremely wide, especially as the proportion of missing units (i.e., $p$) increases. The proposed framework allows for a less conservative view of potential bias that can result from an overlap violation. Finally, the different sensitivity parameters helps theoretically explain the factors that affect bias from overlap violations, which is helpful in considering design decisions to help mitigate overlap violations. 

\paragraph{Remark.} The bias decomposition does not capture cases in which there is a moderating effect from selection itself. In other words, $V_i$ cannot be identical to $S_i$. This could occur if there were a fixed effect from selecting into the experiment that cannot be explained by any of the moderators $X$. In such a setting, $p = 1$ and $R^2_{\tau \sim V}$ = 0. The bias that arises would then be defined as a function of how much larger or smaller the average treatment effect is in the target population, relative to the experimental sample. However, $p=1$ implies that we can no longer leverage information from the experimental sample to inform how large or small the T-PATE is, as \textit{none} of the units in the target population resemble units included in the experimental sample. We provide more discussion in Appendix \ref{app:p1_case}. 

\subsection{Interpreting the Sensitivity Parameters}

\subsubsection{Proportion of Omitted Units} 
The first parameter, $p$, represents the proportion of units in the target population that are not represented by any units in the experimental sample. As $p$ increases towards 1, this implies that a larger portion of units across the target population are excluded, and there is a greater divergence between the transportable population and the target population. 

In many settings, researchers may have a prior for the proportion of units in the target population that have been omitted. This is plausible in settings when researchers have a clear inclusion criteria for the experiment of interest. For example, in the context of the cash transfer program, researchers may be worried about omitting lower educated individuals from the experimental sample. However, as educational attainment is a covariate collected at the population level, researchers can then use the population data to calibrate feasible values of $p$. Furthermore, researchers can often times directly observe overlap violations in the collected covariates. In these settings, it can be useful to fix $p$ to the observed overlap violation, and then perform the sensitivity analysis using just one parameter. We discuss such a setting in Appendix \ref{sec:baseline}.

\subsubsection{Moderating Impact of Omitting Units}
The second sensitivity parameter is an $R^2$ value that represents the amount of treatment effect heterogeneity that can be explained by $V$. More formally, $R^2_{\tau \sim V}$ captures the degree to which $V_i$ moderates the treatment effect. 

In the context of the running example, consider the setting in which $V_i$ represents the set of individuals with little to no familial connections. In particular, individuals with more familial connections in the local government could have a higher chance of having their application pass through the screening process. Then, $R^2_{\tau \sim V}$ would represent the degree to which familial connections moderates the treatment effect. For example, if our outcome of interest is \textit{cash earnings}, we may expect that individuals with a larger number of familial connections to benefit more from receiving a cash transfer, because they will be able to leverage these connections to capitalize on receiving the cash. Individuals who were excluded from the experiment as a result of having a lack of familial connections may see less of a benefit in earnings from receiving the cash. In such a setting, we may expect $R^2_{\tau \sim V}$ to be relatively high (i.e., $R^2_{\tau \sim V} \approx 1$). In contrast, if our outcome of interest is \textit{hours of vocational training}, then we might expect that familial connections does not moderate the treatment effect of receiving a cash transfer by very much. In particular, whether or not individuals are able to receive additional hours of vocational training may not depend on the number of familial connections they have. As a result, in this setting, we would expect $R^2_{\tau \sim V}$ to be relatively low (i.e., $R^2_{\tau \sim V} \approx 0$). 

The magnitude of $R^2_{\tau \sim V}$ will be constrained by how much treatment effect heterogeneity can already be explained by the units in the transportable population (i.e., $\var(\tau_i \mid S_i = 0, V_i = 0)$). The following proposition formalizes. 
\begin{proposition}[Bounds on the Magnitude of $R^2_{\tau \sim V}$] \label{prop:r2_bound}
Consistent with Theorem \ref{thm:bias}, define $R^2_{\tau \sim V}$ as the variation in $\tau_i$ across the target population ($S_i = 0$) that is explained by $V_i$. Then, $R^2_{\tau \sim V}$ is upper bounded by the following: 
\begin{equation} 
R^2_{\tau \sim V} \leq 1- \frac{\overbrace{\var(\tau_i \mid S_i= 0, V_i = 0)}^{\text{Variance in Transportable Pop.}}}{\underbrace{\var(\tau_i \mid S_i = 0)}_{\text{Variance in Target Pop.}}}.
\label{eqn:r2_v}
\end{equation} 
\end{proposition} 
Proposition \ref{prop:r2_bound} restricts the range of possible values that $R^2_{\tau \sim V}$ can take on. In practice, Equation \eqref{eqn:r2_v} cannot be estimated, as it is a function of the total treatment effect heterogeneity across the target population (i.e., $\var(\tau_i \mid S_i = 0)$). However, Equation \eqref{eqn:r2_v} highlights the importance of sampling carefully when recruiting for an experiment. If researchers are able to recruit an experimental sample that capture an adequate amount of moderation in the treatment effect, then the degree to which an overlap violation can impact the estimates is restricted. New recruitment strategies that help researchers consider treatment effect heterogeneity \textit{a priori}, such as \citet{egamidesigning} and \citet{tipton2021beyond}, can help reduce the amount of bias from overlap violations, and improve the external validity of an experimental study.

\subsection{The Role of Treatment Effect Heterogeneity} \label{subsec:var_tau}
Examining the bias from an overlap violation helps reconcile a paradox in the external validity literature. In previous work, authors have shown that as the treatment effect heterogeneity within the experimental sample increases, the overall robustness of an experimental study's external validity appears to decrease \citep[e.g.,][]{huang2022sensitivity, devaux2022quantifying}. This is counterintuitive, as we would expect more homogenous experimental samples to be less generalizable. 

Instead, from Theorem \ref{thm:bias}, we see that recruiting a homogenous sample restricts the set of possible transportable populations that the experiment can feasibly generalize to. Thus, while we may have more robustness in recovering the average treatment effect across the transportable population, we incur a cost in bias from an overlap violation. This trade-off is clear when examining the decomposition of the total treatment effect heterogeneity across the target population (i.e., $\var(\tau_i \mid S_i = 0)$) used in Theorem \ref{thm:bias}. More concretely, Theorem \ref{thm:bias} leverages the following decomposition: 
\begin{equation} 
\var(\tau_i \mid S_i = 0) = \underbrace{\var_w(\tau_i \mid S_i = 1)}_{\substack{\text{Variation captured} \\ \text{by experiment}}} \cdot \overbrace{\frac{1 + p (C_\sigma - 1)}{1-R^2_{\tau \sim V}}}^{\text{Gap from Overlap}}.
\label{eqn:var_tau_decomp0} 
\end{equation} 

The \textit{total} amount of treatment effect heterogeneity is not something within researcher control, and is inherent to the nature of the intervention and the chosen target population. Because $\var(\tau_i \mid S_i = 0)$ is fixed for a specified target population, reducing the variation captured in the experimental sample (i.e., $\var_w(\tau_i \mid S_i = 1)$) will result in an increase in both $p$ and $R^2_{\tau \sim V}$. Thus, while robustness may ostensibly increase with more homogenous experiments, this is only true in cases when the \textit{overall} treatment effect heterogeneity across the target population is also low. 

\subsection{Summary}
To summarize, the bias from an overlap violation can be assessed from two parameters: $p$, which represents the proportion of units omitted in the target population, and $R^2_{\tau \sim V}$, which represents how much treatment effect heterogeneity we have failed to capture in the experimental sample. As such, researchers can vary $\{p, R^2_{\tau \sim V}\}$ and assess how sensitive the estimated T-PATE is to potential overlap violations. We summarize in Figure \ref{fig:summary}. 

Notably, the bias decomposition introduced in Section \ref{sec:bias} rely on parameters that characterize the relationship between $V_i$, $S_i$, and $\tau_i$, \textit{across the target population}. This is because overlap violation considers settings when $V_i$ takes on a constant value across the experimental sample; as such, we must consider the relationship between $V_i$ and $S_i$ and $\tau$ across the target population. This can be inherently challenging to reason about, as it relies on reasoning about information across units we do not observe. In the following section, we will introduce a set of tools that help summarize sensitivity to potential overlap violations, and aid researchers in reasoning about the sensitivity parameters. 

\begin{figure}[!ht]
\noindent\fbox{%
\vspace{2mm}
\parbox{0.975\textwidth}{%
\vspace{1mm}
\singlespacing 
\begin{Step}
    \item Set the hyperparameter $C_\sigma$. In general, we recommend researchers begin by setting $C_\sigma = 1$, and repeat the analysis with varying values of $C_\sigma$. 
    \item Estimate an upper bound for $\var_w(\tau_i \mid S_i = 1)$. We will denote the upper bound as $\overline{\var}_w(\tau_i \mid S_i = 1)$.
    \item Vary $p$ on the range $[0, 1)$.
    \item Vary $R^2_{\tau \sim V}$ on the range $[0, 1)$. 
    \item Evaluate the bias, using Theorem \ref{thm:bias}.
\end{Step}
}
}
\caption{Summary of sensitivity framework.} \label{fig:summary}
\end{figure} 

\section{Sensitivity Summary Tools} \label{sec:tools}

In the following section, we introduce a suite of sensitivity tools to help researchers calibrate their assessment of sensitivity to overlap violations. In particular, we extend sensitivity statistics for routine reporting from the omitted variable bias literature \citep[e.g.,][]{cinelli2020making, huang2022sensitivity, hartman2022sensitivity} for the overlap setting. The proposed summary measures are important in helping researchers assess not only sensitivity to potential violations in overlap, but also the \textit{plausibility} of such violations.

We define a user-specified threshold $b^*$, such that if the T-PATE were to cross $b^*$, this would imply a substantively meaningful change in the research conclusion. If researchers set $b^* = 0$, this implies that we are concerned about an overlap violation that results in the T-PATE being reduced to zero--i.e., there is no impact from the treatment across the target population. Researchers can choose to set different $b^*$ depending on the substantive context. For example, if researchers are worried about a reduction in the magnitude of the T-PATE (e.g., for a cost-benefit analysis for a policy), they may set $b^*$ to be a proportion of the estimated T-PATE. Alternatively, in settings when researchers are worried that an overlap violation results in a change in the statistical significance of the estimated T-PATE, $b^*$ may be set to a value that corresponds to a null result using a percentile bootstrap procedure (see \citealp{huang2022sensitivity} and \citealp{huang2022vbm} for more discussion).

\subsection{Overlap Robustness Value}
We begin by introducing a numerical summary measure, referred to as the \textit{overlap robustness value} (ORV). The ORV provides one way for researchers to quantify how different the estimated treatment effect across the omitted units must be in order for the estimated T-PATE to be reduced to a threshold $b^*$. More formally, the ORV is defined as follows: 
\begin{equation} 
\text{ORV}_{b^*} = \frac{a_{b^*}}{1+a_{b^*}},\ \ \ 
\text{where } a_{b^*} = \sqrt{\frac{(\hat \tau - b^*)^2}{\overline{\var}_w(\tau_i \mid S_i = 1)}}.
\label{eqn:orv} 
\end{equation}

The $\text{ORV}_{b^*}$ represents the minimum amount of variation $V$ must explain in $\tau_i$ across the target population \textit{and} the minimum proportion of units omitted, in order for the estimated treatment effect to be reduced to a threshold $b^*$. Importantly, $\text{ORV}_{b^*}$ is bounded on an interval 0 to 1. As $\text{ORV}_{b^*}$ increases towards 1, this implies that in order for $\hat \tau$ to be reduced to the threshold $b^*$, (1) the majority of the target population units must be excluded from the experimental sample, and (2) the treatment effect across omitted units is substantially different from the estimated treatment effect across the included units. In contrast, if $\text{ORV}_{b^*}$ is relatively small (i.e., close to 0), then this implies that if there are a few units who are omitted from the transportable population, and there exists even a small amount of treatment effect heterogeneity on the dimension of the omitted subgroup, then this would be sufficient to alter our substantive takeaway. 

Like alternative sensitivity summary measures (e.g., the robustness value in \citealp{huang2022sensitivity} and \citealp{cinelli2020making}, design sensitivity in \citealp{huang2023design}, the $e$-value \citealp{ding2016sensitivity}, to name a few), the overlap robustness value provides a convenient, one number summary for how strong the underlying overlap violation must be to overturn our research conclusion. However, $\text{ORV}_{b^*}$ represents only \textit{one} setting in which an overlap violation overturn our research result. For example, it would not capture settings in which an overlap violation results in a large proportion of units in the target population being omitted from the transportable population (i.e., $p$ is large), but the omission does not moderate very much of the treatment effect (i.e., $R^2_{\tau \sim V}$ is low). Similarly, there may be settings when the proportion of units omitted are relatively low (i.e., $p$ is small), but the omitted units may have a treatment effect that is drastically different than the treatment effect across the included units (i.e., $R^2_{\tau \sim V}$ is large). As such, we recommend that researchers not only report the ORV, but also the other sensitivity summary measures to provide a more holistic understanding for the types of overlap violations that can overturn their research conclusions. 

\subsubsection{Illustration on Example}
We calculate the overlap robustness value for both outcomes of interest. For hours of vocational training, we estimate $\text{ORV}_{b^*=0} = 0.37$. This implies that in order for an overlap violation to reduce the estimated T-PATE to zero, 37\% of the units in the target population would have to be omitted from the transportable population, and the omitted subgroup $V_i$ would have to explain 37\% of the variation in the treatment effect heterogeneity. In contrast, for cash earnings, we estimate $\text{ORV}_{b^*=0}=0.18$. This implies that only 18\% of the units in the target population would have to be omitted, and the omitted subgroup $V_i$ would only have to explain 18\% of the variation in the treatment effect heterogeneity in order to reduce the estimated T-PATE to zero.  

\begin{table}[!ht]
\centering 
    \begin{tabular}{lccc}\toprule 
         &  Within-Site & Weighted & ORV$_{b^*=0}$ \\\midrule 
Hours of Vocational Training & 337.46 (16.28) & 207.16 (68.37) & 0.37 \\ 
Cash Earnings &  14.00 (3.76) & 7.22 (7.14) & 0.18 \\ \bottomrule 
    \end{tabular}
    \caption{Overlap robustness values}
\end{table}

\subsection{Bias Contour Plots}
We now introduce a visual summary tool in the form of bias contour plots. Bias contour plots allow researchers to visually assess how the bias from an overlap violation changes as the sensitivity parameters. Along the $y$-axis, we vary the moderating strength of the omitted subgroup (i.e., $R^2_{\tau \sim V}$) from 0 to 1. Along the $x$-axis, we vary the proportion of omitted units $p$ from 0 to 1. 

We also shade in the region of the plot where the overlap violation is strong enough to reduce the T-PATE to $b^*$. How large or small this shaded region is can be a proxy for how much sensitivity there is in the estimated result to potential overlap violations. 

We provide an example of the bias contour plots for both hours of vocational training and cash earnings (see Figure \ref{fig:contours}). Consistent with the overlap robustness values, we see that the shaded region for cash earnings is substantially larger than hours of vocational training, indicating a larger degree of sensitivity to potential overlap violations. In other words, even a small overlap violation could result in changes in our T-PATE estimate. 

\begin{figure}[!ht]
\includegraphics[width=0.49\textwidth]{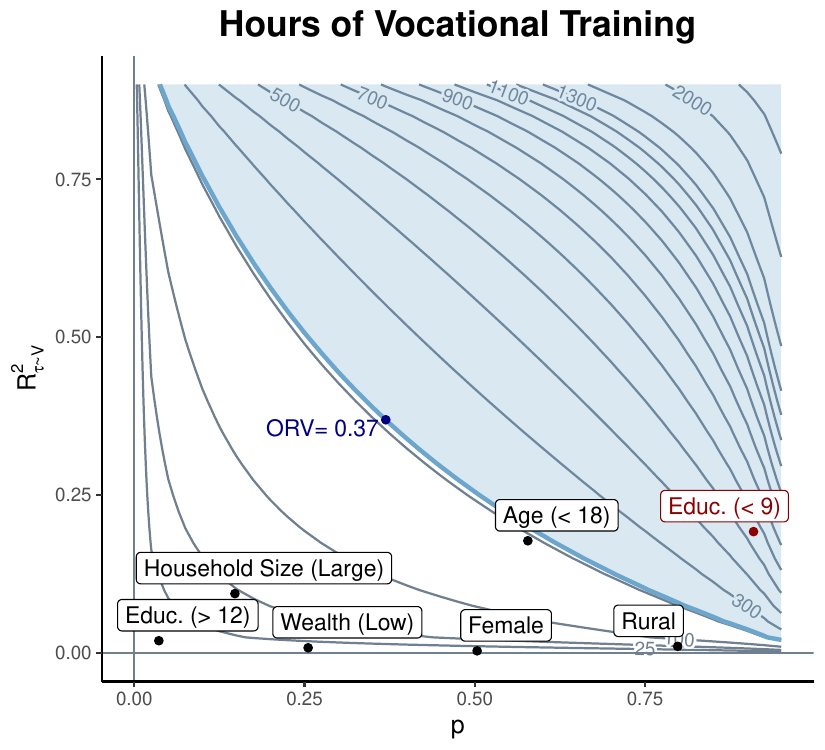}
\includegraphics[width=0.49\textwidth]{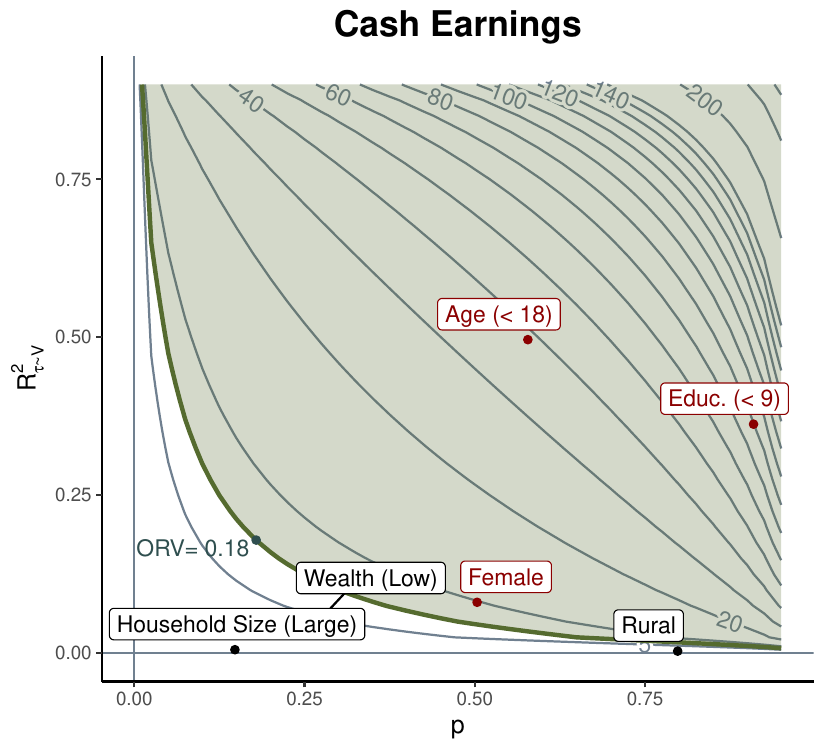}
\caption{Example of Bias Contours, with $C_\sigma = 1$. The shaded region corresponds to the portion of the plot in which the bias is large enough to change the directional sign of the T-PATE estimate. The ORV$_{b^* = 0}$ represents a point along the boundary of the shaded region. We also include the benchmarking results (see Section \ref{sec:benchmark}), which corresponds to the parameters associated with an overlap violation with equivalent strength to omitting the units in the labeled subgroup.}
\label{fig:contours}
\end{figure}

\subsection{Benchmarking}\label{sec:benchmark}
A general challenge in sensitivity analysis is reasoning abut the \textit{plausibility} of bias that is sufficiently large to overturn a research result. While estimating baseline heterogeneity loss provides a helpful reference point for researchers to begin reasoning about heterogeneity loss, it can be difficult to reason about how much additional loss is incurred from an overlap violation. Similarly, the minimum variation explained provides a  helpful summary of how strong an overlap violation must be for our research conclusion to be overturned, but still requires researchers to assess whether or not it is likely that omitting units could result in such an overlap violation. 

To help address these concerns, we introduce benchmarking with observed covariates for the overlap setting. To begin, we define a subgroup $G^{(j)}$, where $j \in \{1, ..., J\}$, and the subgroup $G^{(j)}$ is included in both the experimental sample and the target population. Benchmarking relies on first estimating the relationship between $G_i^{(j)}$ within the experimental sample, and then re-weighting to the target population. Benchmarking allows researchers to estimate the sensitivity parameters associated with an overlap violation that would occur from omitting units, with equivalent strength to omitting the units from $G^{(j)}$. We can equivalently think of benchmarking as defining different transportable populations using the observed subgroups $G^{(j)}$. 

To begin, we introduce a definition of relative strength in overlap violation. 

\begin{definition}[Relative Strength in Overlap Violation]\mbox{}\\
For a subgroup $G^{(j)}$, $j=1, ..., J$, define the relative strength in an overlap violation as follows: 
\begin{align*} 
k_p := \frac{\E(V_i \mid S_i = 0)}{\E( G^{(j)} \mid S_i = 0)} \ \ \ \ \ \ \text{and } \ \ \ \ 
k_\tau := \frac{\cor(\tau_i, V_i \mid S_i = 0)}{\cor(\tau_i, G^{(j)} \mid S_i = 0, V_i = 0)},
\end{align*} 
\end{definition} 
$\{k_p, k_\tau\}$ capture the relationship between the estimated treatment effect heterogeneity and the observed subgroup $G^{(j)}$, relative to the relationship between the individual-level treatment effect and omitted subgroup $V$. If $k_p \geq 1$, this implies that there is a larger proportion of units in the target population missing from the transportable population than if we had omitted the subgroup $G^{(j)}$. If $k_\tau \geq 1$, then this implies that there is a greater degree of moderation from omitting the units from $V$ than from omitting units from an observed subgroup $G^{(j)}$. 

Setting $k_p = k_\tau = 1$, we can estimate the benchmarked sensitivity parameters.

\begin{proposition}[Benchmarked Sensitivity Parameters] \mbox{}\\ 
Assume  $k_p = k_\tau = 1$. Then, under Assumption \ref{assum:cond_ign}, the benchmarked sensitivity parameters can be estimated as follows: 
\begin{align} 
\begin{cases} 
\hat p(j) := \E(G^{(j)} \mid S_i = 0)  \vspace{2mm} \\ 
\hat R^2(j) := \displaystyle \widehat{\cor}_w(\tau_i, G^{(j)} \mid S_i =1)^2.
\end{cases}
\end{align}
\end{proposition} 
$\hat p(j)$ represents the proportion of units excluded from the transportable population, if subgroup $j$ were omitted. $\hat R^2(j)$ represents the amount of moderation that would occur if the units excluded from the transportable population moderated the treatment effect as much as membership in group $G^{(j)}$. $\hat R^2(j)$ is computed by first estimating the relationship between $\tau_i$ and the subgroup membership across the experimental sample, and then re-weighting to the target population. 

The benchmarked parameters simultaneously account for the moderating strength of the omitted subgroup $j$, as well as any correlation that exists between the other covariates $\bX$ and the omitted subgroup. In settings when researchers are interested in evaluating the sensitivity parameters for an overlap violation with greater (or less) strength than an observed subgroup, they can relax the assumption of $k_p = k_\tau = 1$. 

\paragraph{Relative Summary Measures.}
With the benchmarked parameters $\{\hat p(j), \hat R^2(j)\}$, researchers can then estimate the bias that would occur from an overlap violation with equivalent strength to omitting subgroup $G^{(j)}$. A useful measure to also calculate is the \textit{minimum relative overlap bias} that would be needed to reduce the estimated T-PATE to $b^*$: 
$$\text{MROB}(j) = \frac{\hat \tau - b^*}{\text{Bias}(\hat \tau \mid \hat p(j), \hat R^2(j))}.$$
MROB provides another way to quantify how the relative strength of an overlap violation that is needed to overturn a research result. If $\text{MROB}(j) > 1$, this implies that the overlap violation from omitting units outside of the transportable population must be stronger than omitting units from subgroup $G^{(j)}$ to overturn our research result. In contrast, if $\text{MROB}(j) < 1$, then this implies that an overlap violation need not be as strong as the effects from omitting units from a subgroup to overturn the research result. 

\paragraph{Calibrating the Sign of the Bias.}
In practice, the direction of the bias matters. For example, if omitting certain units results in the T-PATE increasing in value, then researchers may be worried that the estimated T-PATE underestimates the true impact of the treatment effect; however, an overlap violation of this nature will not result in a change in our research conclusion. To estimate whether omitting units like an observed subgroup result in an under or overestimation of the T-PATE, researchers can check the sign of $\widehat{\cor}_w(\tau_i, G_i^{(j)} \mid S_i = 1)$. If $\widehat{\cor}_w(\tau_i, G_i^{(j)} \mid S_i = 1) > 0$, then this implies that individuals in the subgroup $G^{(j)}$ have a larger treatment effect than individuals outside of the subgroup. As a result, if the omitted subgroup is similar to the subgroup $G^{(j)}$, we would expect that the overlap violation would result in an underestimation of the estimated treatment effect. In contrast, if $\widehat{\cor}_w(\tau_i, G^{(j)} \mid S_i = 1) < 0$, then this implies that the individuals in subgroup $G^{(j)}$ have a lower treatment effect than individuals outside of the subgroup. As a result, omitting similar units would result in an overestimation of the estimated treatment effect. 

\subsubsection{Illustration on Running Example}
We estimate the benchmarked sensitivity parameters across the running example. We use the observed covariates that are included in the estimated weights to construct the subgroups of interest. The full benchmarking results are provided in Table \ref{tbl:benchmark}. 

From the benchmarking results, we see that certain subgroups result in the largest moderation. In particular, the subgroups with the largest amount of treatment effect heterogeneity are individuals who are young (i.e., individuals younger than 18) and have lower amounts of educational attainment (i.e., years of education is less than 9). 

We see that the minimum relative overlap bias (MROB) values needed to overturn the research result for cash earnings are all relatively lower than the MROB values for hours of vocational training. This is due to two factors. The first is that the estimated T-PATE for cash earnings is inherently lower than the estimated T-PATE for hours of vocational training. As such, less bias is needed to reduce the estimate to zero. However, we also see that the benchmarked $\hat R^2(j)$ values are larger across a few key covariates. In other words, there is a greater degree of moderation that occurs for lower income, young, less educated individuals when receiving a cash transfer on cash earnings, than for hours of vocational training. 

\section{Conclusion} 
In this paper, we introduced a framework for researchers to consider overlap violations when generalizing or transporting their estimated effects. We proposed a novel bias decomposition, which allows researchers to parameterize the bias from an overlap violation with two parameters that are standardized and bounded. Furthermore, the bias decomposition can be flexibly applied across a variety of estimation approaches. We introduced a set of sensitivity tools that researchers can leverage to help understand the plausibility of an overlap violation overturning their research results. These tools are motivated by methods introduced for the omitted variable bias literature, but specifically address the overlap setting. We provide summary measures for researchers to quantify the degree of moderation that must be present from omitting a subset of units, as well as a benchmarking approach, which allows researchers to transparently incorporate their substantive expertise into the sensitivity analysis. 

This work motivates several lines of future research. First, recent work in the observational causal inference literature has introduced the notion of design sensitivity for weighted estimators \citep[e.g.,][]{huang2023design}. Design sensitivity allows researchers to assess different design and estimation decisions \textit{a priori}, with the goal of minimizing sensitivity to potential bias. A natural extension of this paper would be to leverage the proposed sensitivity analysis for overlap violations to introduce a design sensitivity framework for external validity. This would allow researchers to assess the trade-offs they incur in sensitivity to potential bias from using different sampling schemes or experimental designs. Second, the sensitivity framework highlights the importance of understanding, especially at the recruitment stage of an experiment, what covariates moderate the treatment effect. Being able to understand the drivers of treatment effect heterogeneity is crucially important in being able to design an externally valid experiment. 

\clearpage 
\bibliographystyle{chicago} 
\bibliography{bibliography}
\clearpage
\appendix 
\section{Extensions and Additional Discussion} \label{app:extra_details}
\subsection{Outcome Modeling}
Consider the setting in which researchers estimate an outcome model to predict $\tau_i$. Because the model must be estimated across the experimental sample and then used to project into the target population, this implies that at best, researchers will be able to estimate $\tau(X_i) := \E(\tau_i \mid V_i = 0, S_i =1, X_i) = \E(\tau_i \mid S_i =1, X_i)$. Thus, intuitively, the bias from omitting units in $V_i = 1$ from the experimental sample will depend on the relationship between $V_i$ and the part of $\tau_i$ that cannot be explained by $\tau(X_i)$. In particular, we can apply Theorem \ref{thm:bias} to the residuals of the outcome model (i.e., $\tau_i - \tau(X_i)$). The bias will thus depend on how much $V$ can explain the residual variation in $\tau_i - \tau(X_i)$ (i.e., the portion of the individual-level treatment effect that cannot be explained by the outcome model).

While estimating an outcome model can allow researchers to project beyond the convex hull in settings where they have missing covariate support in $X_i$, it does not help account for overlap violations, in which researchers have failed to account for a subgroup in the experimental sample. In particular, because the outcome model must be estimated using the experimental sample data, the outcome model cannot `learn' the relationship between the omitted subgroup and the individual-level treatment effect.

\subsection{Extended Discussion on $C_\sigma$} \label{app:c_sigma}
As described in the main manuscript, $C_\sigma$ controls the relative amount of unmeasured heterogeneity that is present in the omitted units, relative to the included units. One way researchers can calibrate potential values of $C_\sigma$ is to estimate an individual-level treatment effect model across the experimental sample, and use the model to project into the target population. Common examples of models used to estimate treatment effect heterogeneity include regression-based approaches, causal forests \citep{athey2019generalized}, Bayesian Additive Regression Trees (BART) \citep{hill2011bayesian}, etc. (See \citet{kern2016assessing} for more discussion.) With the estimated $\hat \tau(X_i)$, researchers can then compare the variances across different benchmarking subgroups (i.e., $\var(\hat \tau(X_i) \mid G^{(j)})$, where $G^{(j)}$ is defined in Section \ref{sec:benchmark}). 

This can help provide a  useful benchmark to anchor plausibility arguments for $C_\sigma$ values. It is worth cautioning that the calibration will also depend on underlying model used to estimate the individual-level treatment effect. For this reason, we suggest researchers set $C_\sigma = 1$ for minimum reporting, and then vary $C_\sigma$ values. 

\subsubsection{Example: Simulated Data Generating Processes}
To help visualize the impact of changing $C_\sigma$, we numerically simulate some examples. We leverage the population-level linear decomposition introduced in Section \ref{sec:bias} (i.e., Equation \eqref{eqn:tau_lin_decomp}), and set $\hat \alpha = 2$. We simulate $u \sim N(0, \sigma^2(V))$, and define $\sigma^2(V)$ as 
$$\sigma^2(V) := \begin{cases} 1 &\text{if } V = 0\\ 
\sigma^2_1 & \text{if } V = 1.
\end{cases} $$
Because $C_\sigma$ measures the relative variance across the units excluded from the transportable population and we have set $\sigma^2(0) = 1$, $\sigma^2_1$ maps to $C_\sigma$ values. We vary $\sigma^2_1 \in \{0.5, 1, 4\}$, and $\gamma \in \{1.75, 2, 2.5\}$, and set the proportion of omitted units $p = 0.25$. 

Figure \ref{fig:c_sigma_example} visualizes the distribution of $\tau_i$, between the included and excluded groups. We see that for the first scenario (i.e. $C_\sigma = 0.25$), the variation of $\tau_i$ across the omitted units is lower than the included units. In contrast, when $C_\sigma = 4$, there is considerable more variation in the individual-level treatment effect of the omitted units. 

\subsection{Complete Overlap Violation ($p=1$)} \label{app:p1_case} 
In settings in which $p=1$, then this implies that there are \textit{no} units in the target population that are represented by the units in the experimental sample. This means that $V_i = 1$ for all units in the target population. When this is the case, $p = 1$, and $R^2_{\tau \sim V} = 0$. This is because $R^2_{\tau \sim V} = \cor(\tau_i, V_i \mid S_i = 0)^2 = 0$, when $V_i = 1$ for all units $S_i = 0$. As a result, the bias decomposition introduced in Theorem \ref{thm:bias} cannot be directly applied. To evaluate the bias from an overlap violation in which the target population is completely different from the experimental sample, researchers would have to reason about how much larger or smaller the true T-PATE is, relative to the estimated T-PATE. It is worth noting that such a setting would imply that there is a fixed effect from being in the experimental sample, or from the contextual shift between the experimental sample and the target population. If researchers believe this to be the case, it is likely that external validity is limited. 

\subsection{Observable Overlap Violations} \label{sec:baseline}
In the main manuscript, we have assumed that $V_i$ is unknown or latent, in many practical settings, researchers may have \textit{observed} overlap violations. In the following section, we detail how researchers can use available data to estimate a \textit{baseline overlap error}. Furthermore, in settings when we have observed overlap violations, the sensitivity analysis can be simplified to a single parameter sensitivity analysis by fixing $p$ and varying the $R^2_{\tau \sim V}$ value. We propose a summary measure, \textit{minimum variation explained}, which summarizes the threshold variation $V_i$ would have to explain in $\tau_i$ at the target population level, given the baseline overlap error. 

\subsubsection{Estimating Baseline Overlap}
There are often settings in which researchers can directly observe overlap violations between the experimental sample and the target population. In such settings, researchers can directly estimate the proportion of units omitted from the target population. We refer to this as the \textit{baseline overlap error}. We detail two example settings below

\paragraph{Example: Observed Overlap Violations in Pre-Treatment Covariates.} In settings when researchers have access to individual-level demographic data at the target population level, researchers can directly check whether or not there is observable overlap violations across the pre-treatment covariates between the experimental sample and the target population.\footnote{In settings when researchers are using census data, this may not be possible if researchers only have access to summary statistics (i.e., first-order moments) associated with the pre-treatment covariates.} As a simple example, consider the cash transfer setting. The experimental sample consists of individuals between the ages of 14 to 59. In contrast, about 54\% of individuals across the target population are outside this age range. As such, a lower bound for $p$ can naturally be set at 54\%. Researchers can also check for overlap violations in higher order interactions between key covariates. For example, a conservative approach is to perform exact matching between the experimental sample and the target population and assess the proportion of units in the target population cannot be matched to units in the experimental sample. 

\paragraph{Example: Attrition Rate.} In settings when researchers about attrition as a form of overlap violation, they can estimate the attrition rate across the target population, which can serve as a baseline lower bound for $p$. For example, in the cash transfer context, the estimated attrition rate across the target population is 27\%.\\

The baseline overlap error serves as a lower bound for $p$. As such, a low baseline overlap error does not necessarily imply that the true proportion of units omitted from the target population is low, as the overlap violation could occur across latent moderators that are not measured. However, a high baseline overlap error does provide a warning that the true proportion of units omitted from the target population is likely also high. 

\subsubsection{Minimum Variation Explained}
With an estimate of baseline overlap error, researchers can then compute the minimum variation an omitted subgroup must explain in the individual-level treatment effect in order to reduce the T-PATE to $b^*$. We call this the \textit{minimum variation explained} (MVE), formally defined as: 
$$\text{MVE}_{b^*}(p) = \frac{f(p)}{1+f(p)} \text{ where } f(p) = \frac{(\hat \tau - b^*)^2 \cdot \left(1 + C_\sigma \cdot \frac{p}{1-p}\right)}{\var_w(\tau_i \mid S_i = 1) \cdot p }.$$
The $\text{MVE}_{b^*}$ represents the minimum amount of moderation that must occur from omitting the subgroup $V_i$, given the baseline proportion of units omitted, for the T-PATE to be reduced to $b^*$. Like the ORV, the MVE is restricted to a range of 0 to 1.

\section{Proofs and Derivations} \label{app:proofs} 
\subsection{Helpful Lemmas}
\begin{lemma}[Validity of Transportable Weights] \label{lem:unbiased}\mbox{}\\
Define the set of transportable weights (i.e., Equation \eqref{eqn:transportable_weights}): 
\begin{equation} 
w(X_i, V_i = 0) := \frac{P(S_i= 1)}{P(S_i = 0 \mid V_i = 0)} \cdot \frac{P(S_i = 0 \mid X_i, V_i = 0)}{P(S_i =1 \mid X_i, V_i=0)}.
\label{eqn:transportable_weights}
\end{equation} 
Then: 
$$\E(w(X_i, V_i = 0) \cdot \tau_i \mid S_i = 1) = \E(\tau_i \mid V_i = 0, S_i = 0).$$
\end{lemma}
\begin{proof}
We will first show that the weighted estimator will be an unbiased estimator for the average treatment effect, across the transportable population. 
\begin{align*} 
\E&(w(X_i, V_i = 0) \tau_i \mid S_i = 1) \\
&=\sum_{x \in X}  w(X_i = x, V_i = 0) \tau_i P(X_i = x, \tau_i \mid S_i = 1)\\
&=\sum_{x \in X}  w(X_i = x, V_i = 0) \tau_i P(X_i = x, \tau_i \mid S_i = 1, V_i = 0)\\
&= \sum_{x \in X} w(X_i = x, V_i = 0) \tau_i \cdot \frac{P(S_i = 1 \mid X_i = x, \tau_i, V_i = 0)\cdot P(X_i, \tau_i \mid V_i = 0)}{P(S_i = 1, V_i = 0)}
\intertext{By conditional ignorability:}
&= \sum_{x \in X} w(X_i = x, V_i = 0) \tau_i \cdot \frac{P(S_i = 1 \mid X_i = x, V_i = 0)\cdot P(X_i, \tau_i \mid V_i = 0)}{P(S_i = 1)}
\intertext{Substituting in the definition of $w(X_i = x, V_i = 0)$:}
&=\sum_{x \in X} \frac{P(S_i= 1)}{P(S_i = 0 \mid V_i = 0)} \cdot \frac{P(S_i = 0 \mid X_i = x, V_i = 0)}{P(S_i =1 \mid X_i = x, V_i=0)} \tau_i \cdot \frac{P(S_i = 1 \mid X_i = x, V_i = 0)\cdot P(X_i, \tau_i \mid V_i = 0)}{P(S_i = 1)} \\
&= \sum_{x \in X} \tau_i \frac{P(S_i = 0 \mid X_i = x, V_i = 0)}{P(S_i = 0 \mid V_i = 0)} P(X_i, \tau_i \mid V_i = 0)\\
&= \sum_{x \in X} \tau_i \frac{P(S_i = 0 \mid \tau_i, X_i = x, V_i = 0)}{P(S_i = 0 \mid V_i = 0)} P(X_i, \tau_i \mid V_i = 0)\\
&\equiv \E(\tau_i \mid S_i = 0, V_i = 0)
\end{align*} 
In practice, researchers estimate standard inverse propensity score weights $w(X_i)$, and not the transportable weights $w(X_i, V_i = 0)$. However, under the assumption that $X_i \indep V_i$, $w(X_i)$ is proportional to the transportable weights (defined in Equation \eqref{eqn:transportable_weights}:
\begin{align*} 
w(X_i, V_i = 0) &= \frac{P(S_i= 1)}{P(S_i = 0 \mid V_i = 0)} \cdot \frac{P(S_i = 0 \mid X_i, V_i = 0)}{P(S_i =1 \mid X_i, V_i=0)} \propto w(X_i).
\end{align*}

\end{proof}

\begin{lemma}[Validity of Weighted Variance Estimate] \label{lemma:weighted_var} \mbox{}\\
Define $w_i$ as given in Equation \eqref{eqn:transportable_weights}. Furthermore, define $\var_w(A_i)$ as the weighted variance of a variable $A_i$, defined formally as follows: 
$$\var_w(A_i) := \sum_{i:S_i = 1} w_i (A_i - \bar A)^2.$$
Then, $\E[\var_w(A_i)] = \var(A_i \mid S_i = 0, V_i = 0)$. 
\end{lemma} 
\begin{proof} 
The results of this lemma follow immediately from Lemma \ref{lem:unbiased}.
\end{proof} 
\subsection{Proof of Theorem \ref{thm:bias}}
\begin{proof} 

Using Definition \ref{def:overlap}, we can generally decompose the error in recovering the T-PATE into two parts: 
\begin{align} 
\begin{split}
\E(\hat \tau)-\tau &=\underbrace{\E(\hat \tau) - \E(Y_i(1) - Y_i(0) \mid S_i = 0, V_i = 0)}_{\text{(a) Error in Recovering ATE across Transportable Population}} \\
&~~~~+\underbrace{\E(Y_i(1) - Y_i(0) \mid S_i = 0, V_i = 0) - \E(Y_i(1) - Y_i(0) \mid S_i = 0)}_{\text{(b) Gap between Transportable and Target Population}}. 
\end{split}
\label{eqn:error}
\end{align}

The error from generalizing (or transporting) experimental results will arise from (1) error from recovering the average treatment effect across the transportable population (i.e., Equation \ref{eqn:error}-(a)), and (2) error from slippage between the transportable and target population. By assumption, the first source of error (Equation \eqref{eqn:error}-(a)) is zero. Thus, we will focus on the second error (Equation \eqref{eqn:error}-(b)).

From Equation \eqref{eqn:error}-(b): 
\begin{equation} \underbrace{\big( \E(Y_i(1) - Y_i(0) \mid S_i = 0, V_i = 0) - \E(Y_i(1) - Y_i(0) \mid S_i = 0, V_i = 1) \big)}_{(*)} \cdot P(V_i = 1 \mid S_i = 0)
\label{eqn:overlap_error}
\end{equation}
The derivation proceeds in two parts. In the first part, we will re-write Equation \ref{eqn:overlap_error} as a function of an $R^2$ value (i.e., $R^2_{\tau \sim V}$), the proportion of units omitted (represented by $p$), and the variance of $\tau_i$ across the target population. In the second part, we will show that the variance of $\tau_i$ across $S_i = 0$ can be re-written as a function of the treatment effect heterogeniety across the experimental sample (i.e., $\var(\tau_i \mid S_i = 1)$) and the $R^2$ value. 

\noindent \underline{\textbf{Part 1.}} To begin, we note the following:
\begin{align} 
\cov&(\tau_i, V_i \mid S_i = 0) \nonumber \\
=& \E(\tau_i V_i \mid S_i = 0) - \E(\tau_i \mid S_i = 0) \E(V_i \mid S_i = 0) \nonumber \\
=& \E(\tau_i V_i \mid S_i = 0, V_i = 1) \E(V_i\mid S_i = 0) - \E(\tau_i \mid S_i = 0) \E(V_i \mid S_i = 0) \nonumber \\
=& \E(\tau_i V_i \mid S_i = 0, V_i = 1) \E(V_i \mid S_i = 0)  \nonumber \\
&- \big(\E(\tau_i \mid S_i = 0, V_i = 1) \E(V_i \mid S_i = 0)^2 + \E(\tau_i \mid S_i = 0, V_i = 0) \E(V_i \mid S_i = 0) \E(1-V_i \mid S_i = 0) \big) \nonumber \\
=&\big( \E(\tau_i \mid S_i = 0, V_i = 1) - \E(\tau_i \mid S_i = 0, V_i = 0)\big) \var(V_i \mid S_i = 0)\nonumber 
\intertext{Thus, we can re-arrange the terms:}
\implies& \E(\tau_i \mid S_i = 0, V_i = 1) - \E(\tau_i \mid S_i = 0, V_i = 0) = \frac{\cov(\tau_i, V_i \mid S_i = 0)}{\var(V_i \mid S_i = 0)}
\label{eqn:gamma} 
\end{align} 
A useful duality that we will exploit is that we can equivalently consider the following linear decomposition of $\tau_i$: 
\begin{equation} 
\tau_i = \hat \alpha + \hat \gamma V_i + u_i,
\end{equation} 
where $\hat \alpha := \E(\tau_i \mid V_i = 0, S_i = 0)$, and $\hat \gamma$ is equivalent to Equation \eqref{eqn:gamma}. 

Then, we can rewrite Equation \eqref{eqn:overlap_error}-(a) as: 
\begin{align} 
\E&(Y_i(1) - Y_i(0) \mid S_i = 0, V_i = 0) - \E(Y_i(1) - Y_i(0) \mid S_i = 0, V_i = 1)\nonumber \\ 
&= \frac{\cov(\tau_i, V_i \mid S_i = 0)}{\var(V_i \mid S_i = 0)}\nonumber \\
&= \cor(\tau_i, V_i \mid S_i = 0) \cdot \sqrt{\frac{\var(\tau_i \mid S_i = 0)}{\var(V_i \mid S_i = 0)}}\nonumber \\
&= R_{\tau \sim V} \sqrt{\frac{\var(\tau_i \mid S_i = 0)}{p(1-p)}},
\label{eqn:bias_decomp_1}
\end{align} 
where we have defined $R^2_{\tau \sim V} := \cor(\tau_i, V_i \mid S_i = 0)^2$, and $p := \E(V_i \mid S_i = 0)$.\\

\noindent \underline{\textbf{Part 2.}} We will now show that $\var(\tau_i \mid S_i = 0)$ can be decomposed as $\frac{C_\sigma \cdot \var(\tau_i \mid S_i = 1)}{1-R^2_{\tau \sim V}}$.

First, we can re-write the variance of the noise term $u_i := \tau_i - \hat \alpha - \hat \gamma V_i$ as: 
\begin{align*}
\var(u_i \mid S_i = 0) =& \var(u_i \mid V_i = 1, S_i = 0) P(V_i = 1 \mid S_i = 0) + \var(u_i \mid V_i = 0\mid S_i = 0) P(V_i = 0\mid S_i = 0) + \\
&\E(u_i \mid V_i = 1, S_i = 0)^2 \var(V_i \mid S_i = 0) + \E(u_i \mid V_i = 0, S_i = 0)^2 \var(V_i \mid S_i = 0)\\
&- 2 \E(u_i \mid V_i = 1, S_i = 0) P(V_i = 1 \mid S_i = 0) \E(u_i \mid V_i = 0, S_i = 0) P(V_i = 0\mid S_i = 0) \\
=& \var(u_i \mid V_i = 1, S_i = 0) P(V_i = 1 \mid S_i = 0) + \var(u_i \mid V_i = 0\mid S_i = 0) P(V_i = 0\mid S_i = 0) + \\
&\underbrace{\left(\E(u_i \mid V_i = 1, S_i = 0) -\E(u_i \mid V_i = 0, S_i = 0)\right)^2}_{\equiv 0}\cdot \var(V_i \mid S_i = 0)\\
=& \var(u_i \mid V_i = 1, S_i = 0) P(V_i = 1 \mid S_i = 0) + \var(u_i \mid V_i = 0\mid S_i = 0) P(V_i = 0\mid S_i = 0) \\
=& \var(u_i \mid V_i = 1, S_i = 0) p + \var(u_i \mid V_i = 0\mid S_i = 0) (1-p),
\end{align*} 
which follows from Law of Total Variance.

This allows us to then decompose $\var(\tau_i \mid S_i = 0)$ as: 
\begin{align*}
\var(\tau_i \mid S_i = 0) &= \var(u_i \mid S_i = 0) + \var(\hat \gamma V_i \mid S_i = 0) \\
&= \var(u_i \mid V_i = 0, S_i = 0) (1-p) + \var(u_i \mid V_i = 1, S_i = 0) (1-p) + \var(\hat \gamma V_i \mid S_i = 0) 
\end{align*} 
Re-arranging the terms: 
\begin{align} 
1-R^2_{\tau \sim V} &= \frac{\var(u_i \mid V_i = 0, S_i =0) \cdot (1-p) + p \cdot \var(u_i \mid V_i = 1, S_i = 0)}{\var(\tau_i \mid S_i = 0)}\nonumber 
\intertext{Note that $\var(\tau_i \mid V_i = 0, S_i = 0) \equiv \var(u_i \mid V_i =0, S_i = 0)$. Furthermore, define $C_\sigma := \var(u_i \mid V_i = 1, S_i = 0)/\var(u_i \mid V_i = 0, S_i = 0)$. Then:}
&= \frac{\var(\tau_i \mid V_i = 0, S_i = 0) (1-p(1 - C_\sigma))}{\var(\tau_i \mid S_i =0)} \label{eqn:r2_decomp}
\intertext{By Lemma \ref{lemma:weighted_var}, $\var_w(\tau_i \mid S_i = 1) = \var(\tau_i \mid S_i = 0, V_i = 0)$:}
&= \frac{\var_w(\tau_i \mid S_i = 1) (1-p(1 - C_\sigma))}{\var(\tau_i \mid S_i =0)}\nonumber 
\end{align} 
Re-arranging the terms provides us with an expression of $\var(\tau_i \mid S_i = 0)$: 
\begin{equation} 
\var(\tau_i \mid S_i = 0)= \frac{(1-p(1-C_\sigma)) \cdot \var_w(\tau_i \mid S_i = 1)}{1-R^2_{\tau \sim V}}
\label{eqn:var_tau_decomp}
\end{equation} 

Substituting Equation \eqref{eqn:bias_decomp_1} and \eqref{eqn:var_tau_decomp} into Equation \eqref{eqn:overlap_error}, we arrive at Equation \ref{eqn:bias}: 
\begin{align*} 
\text{Bias}(\hat \tau) = \sqrt{\frac{R^2_{\tau \sim V}}{1-R^2_{\tau \sim V}} \cdot \left( 1 + C_\sigma \cdot \frac{p}{1-p}\right) p \cdot \var_w(\tau_i \mid S_i = 1)}.
\end{align*} 
\end{proof} 

\subsection{Proof of Proposition \ref{prop:r2_bound}}
Recall from Equation \eqref{eqn:r2_decomp}, $R^2_{\tau \sim V}$ can be written as follows: 
\begin{align*} 
R^2_{\tau \sim V} = 1 - \frac{\var(\tau_i \mid V_i = 0, S_i = 0) (1- p (1-C_\sigma))}{\var(\tau_i \mid S_i = 0)}
\end{align*} 
Then, since $C_\sigma \geq 0$, and $p \in [0,1]$, it follows that $R^2_{\tau \sim V}$ is upper bounded by: 
\begin{align*} 
R^2_{\tau \sim V} \leq 1 - \frac{\var(\tau_i \mid V_i = 0, S_i = 0)}{\var(\tau_i \mid S_i = 0)},
\end{align*} 
which is reached when $p = 0$. 

\subsection{Derivation of ORV}
The $\text{ORV}_{b^*}$ is defined as the minimum amount of variation $V$ must explain in $\tau_i$ across the target population \textit{and} the minimum proportion of units omitted, in order for the estimated treatment effect to be reduced to a threshold $b^*$. The derive the $\text{ORV}_{b^*}$, we begin with the bias formula: 
\begin{align*} 
\text{Bias}(\hat \tau) = \sqrt{\frac{R^2_{\tau \sim V}}{1-R^2_{\tau \sim V}} \cdot \left( 1 + C_\sigma \cdot \frac{p}{1-p}\right) p \cdot \var_w(\tau_i \mid S_i = 1)}.
\end{align*} 
Then, for the bias to be large enough to reduce the T-PATE to $b^*$, $\text{Bias}(\hat \tau) = \hat \tau - b^*$. We set $C_\sigma = 1$, and $R^2_{\tau \sim V} = p = \text{ORV}_{b^*}$:
\begin{align*}
\hat \tau - b^* =\sqrt{\frac{\text{ORV}_{b^*}}{1-\text{ORV}_{b^*}} \cdot  \frac{\text{ORV}_{b^*}}{1-\text{ORV}_{b^*}} \cdot \var_w(\tau_i \mid S_i = 1)}.
\end{align*} 
Squaring both sides and directly solving for $\text{ORV}_{b^*}$ recovers Equation \eqref{eqn:orv}.

\section{Additional Tables and Results}

\begin{table}[ht]
\footnotesize
\centering
\begin{tabular}{lcccccccccc}
  \toprule
  & \multicolumn{5}{c}{Hours of Vocational Training} & \multicolumn{5}{c}{Cash Earnings} \\ 
  \cmidrule(lr){2-6} \cmidrule(lr){7-11}
Covariate & $\hat R^2(j)$ & $\hat p(j)$ & Bias & MROB & Sign &  $\hat R^2(j)$ & $\hat p(j)$ & Bias & MROB & Sign\\ 
  \midrule
  Years of Education \\ 
$\quad$ Under 9 & 0.19 & 0.91 & 546.66 & 0.38 & $-$ & 0.36 & 0.91 & 79.12 & 0.09 & $-$ \\ 
$\quad$ Over 12 & 0.02 & 0.04 & 9.69 & 21.37 & + & 0.01 & 0.04 & 0.66 & 10.96 & + \\ 
Wealth\\
  $\quad$ Bottom 25\% & 0.01 & 0.26 & 19.09 & 10.85 & + & 0.03 & 0.26 & 3.58 & 2.02 & + \\ 
  $\quad$ Top 25\% & 0.02 & 0.25 & 29.67 & 6.98 & + & 0.02 & 0.25 & 2.42 & 2.98 & + \\ 
  Geographic Region \\ 
  $\quad$ Rural & 0.01 & 0.80 & 73.36 & 2.82 & $-$ & 0.00 & 0.80 & 3.71 & 1.95 & $-$ \\ 
  $\quad$ Urban & 0.01 & 0.20 & 18.58 & 11.15 & + & 0.00 & 0.20 & 0.94 & 7.69 & + \\ 
  Demographic Information \\ 
  $\quad$ Female & 0.00 & 0.50 & 21.52 & 9.62 & + & 0.08 & 0.50 & 9.87 & 0.73 & + \\ 
  $\quad$ Small Household & 0.09 & 0.15 & 47.51 & 4.36 & $-$ & 0.01 & 0.15 & 1.02 & 7.10 & + \\ 
  $\quad$ Large Household & 0.06 & 0.11 & 31.85 & 6.50 & + & 0.05 & 0.11 & 2.59 & 2.79 & $-$ \\ 
  Age Buckets\\
  $\quad$ Under 18 & 0.18 & 0.58 & 192.53 & 1.08 & $-$ & 0.50 & 0.58 & 38.49 & 0.19 & $-$ \\ 
  $\quad$ 18 to 25 & 0.02 & 0.13 & 20.44 & 10.14 & + & 0.71 & 0.13 & 19.57 & 0.37 & + \\ 
  $\quad$ 25 to 35 & 0.25 & 0.13 & 78.10 & 2.65 & + & 0.08 & 0.13 & 3.75 & 1.93 & + \\ 
  $\quad$ 36 to 50 & 0.00 & 0.09 & 4.20 & 49.31 & + & 0.00 & 0.09 & 0.38 & 18.90 & + \\ 
  $\quad$ 50 and above & 0.00 & 0.08 & 0.30 & 687.29 & + & 0.00 & 0.08 & 0.04 & 191.23 & + \\ 
  Interactions \\ 
  $\quad$ Rural \\
  $\quad$ $\quad$ $\times$ Bottom 25\% Wealth & 0.01 & 0.23 & 15.37 & 13.48 & + & 0.02 & 0.23 & 2.52 & 2.87 & + \\ 
  $\quad$ $\quad$ $\times$ $<$ 9 Yrs of Educ. & 0.06 & 0.74 & 151.90 & 1.36 & $-$ & 0.01 & 0.74 & 5.24 & 1.38 & $-$ \\ 
  $\quad$ $\quad$ $\times$ Under 18 & 0.12 & 0.46 & 123.04 & 1.68 & $-$ & 0.14 & 0.46 & 12.61 & 0.57 & $-$ \\ 
  $\quad$ Urban \\
  $\quad$ $\quad$ $\times$ Bottom 25\% Wealth & 0.00 & 0.03 & 2.39 & 86.79 & + & 0.02 & 0.03 & 0.72 & 10.06 & + \\ 
  $\quad$ $\quad$ $\times$ $<$9 Yrs of Educ. & 0.00 & 0.17 & 3.79 & 54.72 & $-$ & 0.10 & 0.17 & 4.82 & 1.50 & $-$ \\ 
  $\quad$ $\quad$ $\times$ Under 18 & 0.01 & 0.11 & 10.92 & 18.96 & $-$ & 0.21 & 0.11 & 6.08 & 1.19 & $-$ \\
   \bottomrule
\end{tabular}
\caption{Benchmarking Results. Sign refers to the direction that the estimated T-PATE would move if we omitted a subgroup as strong as the benchmarked group. For example, all $+$ groups indicate that the T-PATE would increase, if we included additional units. In contrast, all $-$ groups indicate the T-PATE would decrease.}
\label{tbl:benchmark}
\end{table}
\end{document}